\documentclass[a4paper]{article}

\usepackage[T1]{fontenc}
\usepackage[utf8]{inputenc}
\usepackage{lmodern}
\usepackage{microtype}

\usepackage{booktabs}
\usepackage{nicefrac}
\usepackage{siunitx}
\usepackage{xcolor}
\usepackage{todonotes}
\usepackage{svg}
\usepackage{amsmath,amssymb,amsthm}
\usepackage[hidelinks]{hyperref}

\usepackage{subcaption}
\usepackage{cleveref}

\definecolor{myred}{rgb}{0.55, 0.0, 0.0}
\definecolor{mygreen}{rgb}{0.0, 0.42, 0.24}
\definecolor{myblue}{rgb}{0.02, 0.05, 0.6}
\definecolor{myyellow}{rgb}{0.8, 0.6, 0.2}
\definecolor{myteal}{rgb}{0.2, 0.3, 0.3}

\usepackage{tcolorbox}
\usepackage{xspace}
\usepackage{xargs}
\usepackage{framed}

\usepackage{tikz}
\usetikzlibrary{arrows}
\usetikzlibrary{arrows.meta}
\usetikzlibrary{calc}
\usetikzlibrary{decorations.pathreplacing}
\usetikzlibrary{math}
\pgfdeclarelayer{bg}
\pgfsetlayers{bg,main}

\usepackage{multirow}
\usepackage{makecell}

\newenvironment{tightcenter}
  {\parskip=0pt\par\nopagebreak\centering}
  {\par\noindent\ignorespacesafterend}

\newlength{\RoundedBoxWidth}
\newsavebox{\GrayRoundedBox}
\newenvironment{GrayBox}[1]{\setlength{\RoundedBoxWidth}{\linewidth-4.5ex}
\def\boxheading{#1}
\begin{lrbox}{\GrayRoundedBox}
\begin{minipage}{\RoundedBoxWidth}}{\end{minipage}
\end{lrbox}\begin{tightcenter}\begin{tikzpicture}\node(Text)[draw=black!20,fill=white,rounded corners,inner sep=2ex,text width=\RoundedBoxWidth]{\usebox{\GrayRoundedBox}};
\coordinate(x) at (current bounding box.north west);
\node [draw=white,rectangle,inner sep=3pt,anchor=north west,fill=white]
at ($(x)+(10.5pt,.75em)$) {\boxheading};
\end{tikzpicture}
\end{tightcenter}\vspace{0pt}\ignorespacesafterend
}

\newenvironment{problembox}[1]{\noindent\ignorespaces \FrameSep=6pt\parindent=0pt\vspace*{.5em}
\begin{GrayBox}{\textsc{#1}}\newcommand\Input{Input:}\newcommand\Prob{Problem:}\begin{tabular*}
    {\columnwidth}
    {@{\hspace{-0.5em}} >{\itshape} p{4em} p{0.8\columnwidth} @{}}
}{
\end{tabular*}\end{GrayBox}\vspace*{-1.0em}
\ignorespacesafterend
}
 \newcommand{\noopsort}[2]{#2}

\title{Planar Network Diversion}

\author{Matthias Bentert\thanks{University of Bergen, \texttt{Matthias.Bentert@uib.no}}
\and
Pål Grønås Drange\thanks{University of Bergen, \texttt{Pal.Drange@uib.no}}
\and
Fedor V.\ Fomin\thanks{University of Bergen, \texttt{Fedor.Fomin@uib.no}}
\and
Steinar Simonnes\thanks{University of Bergen}
}

\newtheorem{theorem}{Theorem}

\newtheorem{lemma}{Lemma}
\newtheorem{proposition}{Proposition}
\newtheorem{observation}{Observation}
\newtheorem{corollary}{Corollary}

\newcommand{\pname}[1]{\textsc{#1}}

\DeclareMathOperator{\cut}{cut}

\begin{document}

\maketitle

\begin{abstract}
    \textsc{Network Diversion} is a graph problem that has been extensively studied in both the network-analysis and operations-research communities as a measure of how robust a network is against adversarial disruption.
    This problem is especially well motivated in transportation networks, which are often assumed to be planar.
    Motivated by this and recent theoretical advances for \textsc{Network Diversion} on planar input graphs, we develop a fast~$O(n \log n)$ time algorithm and present a practical implementation of this algorithm that is able to solve instances with millions of vertices in a matter of seconds.
\end{abstract}

\section{Introduction}
\label{sec:intro}
The \pname{Network Diversion} problem is a variant of the classic \pname{Minimum Cut} problem, which models network vulnerability to adversarial disruption.
In this problem, we are given two vertices~$s$ and~$t$ and an edge~$b$ in a graph. The task is to identify at most~$k$ edges such that, after their removal, every $s$-$t$ path is forced to pass through~$b$, while ensuring that at least one such path remains.
Equivalently, the problem can be reformulated as finding a minimal~$s$-$t$ cut of size at most~$k+1$ that includes the edge~$b$.
Although \pname{Network Diversion} may initially appear to be a minor variation of \pname{Minimum Cut}, the two problems differ significantly in terms of computational complexity.
While \pname{Minimum Cut} is solvable in polynomial time, \pname{Network Diversion} is NP-hard on directed graphs. For undirected graphs, however, it remains an open question whether \pname{Network Diversion} is NP-hard or admits a polynomial-time algorithm.

In the context of transportation networks, which are often modeled as planar, undirected, and weighted graphs, the ability to compute diverse \( s \)-\( t \)-cuts is crucial for assessing network resilience and planning infrastructure. Unlike traditional minimum-cut computations, the \pname{Network Diversion} framework allows exploring different minimal cuts by enforcing the inclusion of a specific edge~\( b \).
We show empirically that this allows us to find a set of minimal cuts that is quite diverse.
Given that planar graphs have $O(n)$ edges, we can efficiently enumerate a diverse set of potential $s$-$t$-cuts in \( O(n^2 \log n) \) time, providing a tractable means to identify critical bottlenecks.  (See \Cref{fig:diverse-cuts} in the appendix for an illustration.) Moreover, when edge weights represent failure probabilities, the problem of determining whether a given edge~\( b \) is likely to become an \( s \)-\( t \)-bridge reduces to a minimum-weight cut problem in a transformed graph with logarithmic weights. This formulation allows for efficient computation of the probability that all alternative paths fail, making it a valuable tool for transportation planning, security analysis, and network robustness evaluation.

\bigskip

\textsc{Network Diversion} has garnered attention from both the operations-research and network-analysis communities~\cite{cintron-arias2001networkdiversion,
curet2001networkdiversion,
erken2002branchandboundalgorithm,
lee2019combinatorialbenders,
phillips1993networkinhibition,
wood1993deterministicnetwork}.
It measures the vulnerability of a network against flow-manipulation by sabotaging connections in the network and is formally defined as follows.
\begin{problembox}{Planar Network Diversion}
  \Input & An undirected planar graph $G=(V,E)$, an edge-weight function denoting the destruction cost~$w \colon E \rightarrow \mathbb{R}_{\geq 0}$, two vertices~$s$ and~$t$, an edge~$b$, and a budget~$k \in \mathbb{R}_{\geq 0}$.\\
  \Prob  & Does there exist a set~$F \subseteq E$ of edges of total weight at most~$k$ such that $b$ is an~$s$-$t$-bridge in $G-F$?\\
\end{problembox}

In particular, if there is a solution to the problem with small~$k$, then the network is vulnerable as an adversary can damage only few edges to divert all flow or traffic between~$s$ and~$t$ in the network to a single target connection~$b$ (which can be chosen to be particularly susceptible within the network).
Suppose an adversary wants to attack a communication network by forcing data to be rerouted through specific links.
The problem is to determine the smallest number of edges that need to be removed to successfully carry out this attack, providing a measure of the network's vulnerability.

The problem is particularly relevant to
transportation networks and has been extensively studied in the context
of planar graphs.  Cullenbine et al.~\cite{cullenbine2011newresults}
provided a polynomial-time algorithm for \pname{Network Diversion} on
planar graphs under the condition that both terminals,~$s$ and~$t$, are
located on the outer face~\cite{cullenbine2013theoreticalcomputational}.
They raised the question of whether this algorithm could be extended to
arbitrary planar graphs, which has been answered in the
affirmative by Bentert et al.~\cite{bentert2024twosets}.
We give a new algorithm here which is deterministic, conceptually simpler, and faster.

\subsection{Our Results}

\noindent
We present the first deterministic algorithm for \pname{Network Diversion} on weighted planar graphs that runs in truly polynomial time. Specifically, it achieves a running time in~$O(n \log n)$ with only small hidden constants, making it highly efficient in practice.

Our algorithm is, as previous algorithms, based on the correspondence between cuts and
cycles in dual graphs.  Specifically, we use this relationship to find a
path from the face on the left-hand side of edge~\(b\) to its right-hand
side, which is completed to a cycle by adding the dual edge~$b^\star$.  This path
ensures that vertex~\(s\) is positioned on one side of the cycle and
vertex~\(t\) on the opposite side, achieved by identifying an odd-length
path~\cite{derigs1985efficientdijkstralike,lapaugh1984evenpathproblem}
in a transformed version of the dual graph.
Using methods to find the least expensive odd-length paths in weighted graphs, we efficiently determine the most cost-effective minimal cut that includes \(b\).
This approach has potential implications for designing other
types of cuts in planar graphs, applicable even when edges are labeled
with elements from various
groups~\cite{iwata2022findingshortest,kobayashi2017findingshortest}.

\smallskip

As part of our algorithm, we implement Derig's algorithm to find the least expensive odd-length $s$-$t$-path and conduct experiments with two different data structures.  We believe this implementation to be of independent interest.

\bigskip
\noindent
\textbf{Cut--Uncut Problems.}
\pname{Network Diversion} belongs to the broader category of \emph{cut--uncut problems}~\cite{gray2012removinglocal,vanthof2009partitioninggraphs}.
There, the goal is to separate specific terminals while maintaining
connectivity between designated terminal pairs.  These problems are notoriously hard and usually NP-hard.
In the context of \pname{Network Diversion}, we are able to leverage the small number of terminals and the planarity of the input graph to derive a polynomial-time algorithm.
In a recent paper, Bentert et al.~\cite{bentert2024twosets} showed that
\pname{Two--Sets Cut--Uncut} is fixed-parameter tractable in
the size of~$S$ and~$T$ on planar graphs.
In this problem, we are given a planar graph, and two sets~$S$ and~$T$ of vertices, and the goal is to find a small set of edges that separates all
of~$S$ from all of~$T$, while maintaining connectivity within $S$ and~$T$.
This result immediately implies a polynomial-time algorithm for \textsc{Network Diversion} on planar input graphs as shown next.
Let~$b = \{b_s,b_t\}$.
Consider a solution~$F$ and the graph~$G' = G - (F \cup \{b\})$.
Note that exactly one endpoint of~$b$ is in the same connected component of~$G'$ as~$s$ and the other endpoint is in the connected component of~$t$.
We can guess which endpoint is in the connected component of~$s$ (let us assume without loss of generality~$b_s$) and then solve the instance of \textsc{Two--Sets Cut--Uncut} with~$S=\{s,b_s\}$ and~$T=\{t,b_t\}$.
We mention in passing that their algorithm cannot handle edge weights and that their algorithm is randomized.
The algorithm by Bentert et al.\ also employs an algebraic approach, which differs fundamentally from the graph search algorithm we present herein and that we believe to be less competitive.
We overcome both of these issues and present a simple, deterministic algorithm that can handle edge weights.
Our algorithm is implemented in Rust and available at an anonymized repository~\cite{code2024github}.

\subsection{Background}

In 1993, Phillips~\cite{phillips1993networkinhibition} introduced
\textsc{Network Inhibition} as a flow-interdiction problem as follows.
The input consists of a graph~$G$, two vertices~$s$ and~$t$, a budget~$k$, edge capacities~$c$ and \emph{edge destruction costs}~$d$ (i.e., each edge $e$ has a capacity $c_e$ and a cost~$d_e$ for destroying~$e$).
One can now assign each edge~$e$ a value~$\alpha_e$.
This means that one pays $\alpha d_e$ to reduce the capacity of~$e$ to~$(1-\alpha) c_e$.
The task is to assign numbers in such a way that the maximum flow between~$s$ and~$t$ is minimized (with respect to the new edge capacities) while the cost of all modifications does not exceed the budget~$k$.
Phillips proved NP-completeness on subcubic graphs, and weak NP-completeness for series-parallel graphs, planar graphs, bandwidth-3 graphs, and more~\cite{phillips1993networkinhibition}.
They complemented their findings by providing polynomial-time algorithms for outerplanar graphs and an FPTAS for planar graphs.

That problem was further studied by
Wood~\cite{wood1993deterministicnetwork}, who gave several other NP-completeness results and implemented the problem as an ILP and considered different LP relaxations.
We refer to the doctoral thesis of
Kallemyn~\cite{kallemyn2015modelingnetwork} for a thorough survey on the
subject of modeling network interdiction tasks.

Curet~\cite{curet2001networkdiversion} introduced the problem \textsc{Network Diversion} on directed graphs.
They showed that the problem is NP-complete and implemented an ILP for
solving the problem on real-world instances.
Cintron-Arias et al.~\cite{cintron-arias2001networkdiversion}
used the Lagrangean relaxation
for integer programming to implement ILP-based
heuristic algorithms for real-world networks.

A decade later, Cullenbine, Wood, and
Newman~\cite{cullenbine2013theoreticalcomputational} studied the problem
on undirected graphs.
They showed how to solve \textsc{Network Diversion} on so-called $s$-$t$-planar graphs in polynomial time in both the directed and undirected setting.
An \emph{$s$-$t$-planar graph} is a planar graph where~$s$ and~$t$ belong to a common face.
Their algorithm is based on the following observation:
If we find a path~$P$ in~$G$ that contains $b$\footnote{We can find an $s$-$t$-path containing~$b$ using an algorithm for two-disjoint paths, or using the Odd Path algorithm (see Section~\ref{sec:anapplication}).}, then a cycle in the dual, containing~$b^\star$ separates~$s$ from~$t$ if and only if $P$ is \emph{crossed} an odd number of times.  They refer to $P$ as a \emph{reference path}. We use the same idea and explain it in details in the next section.

\begin{proposition}[Cullenbine et al.~\cite{cullenbine2013theoreticalcomputational}]
  A solution to \textsc{Network Diversion} on a
  planar graph~$G$ corresponds directly to a minimum-weight, simple,
  odd-parity cycle $E^\star_C \supseteq \{ {b}^{*} \}$ in the dual graph of~$G$, where parity is measured with respect to a simple $(s,t)$-reference path in~$G$ that
  contains $b$.
\end{proposition}

They also note the obstacle for solving \textsc{Network Diversion} on planar graphs.

\begin{quotation}
  \it
  {\Large``}It is easy to identify a reference path as required by the proposition if one exists.
  And then, using shortest-path techniques, it is easy to find a minimum-weight odd-parity cycle $E^\star_S$ in~$G^\star$
  such that $b^\star \in E^\star_S$.  Unfortunately, such an approach does not lead to a general, efficient method for solving \textsc{Network Diversion} on undirected planar graphs, because $E^\star_S$ may not be simple, and because minimality of the corresponding $(s,t)$-cut demands a simple cycle.{\Large''}
\end{quotation}

\noindent
We will demonstrate in \Cref{sec:network-diversion} how we circumvent this obstacle.
We note that the question of whether \textsc{Network Diversion} on general undirected graphs is polynomial-time solvable or NP-hard remains open.

Cullenbine et al.~\cite{cullenbine2013theoreticalcomputational} also ran experiments on both generated and
real-world data sets using a Mixed-Integer-Linear-Programming (MILP) approach.
We compare our implementation to theirs and find that our algorithm is faster by orders of magnitude.

\section{Preliminaries and Notation}
\label{sec:prelim}

Usually, when working with dual graphs, one needs to consider multigraphs.
However, we will only be computing simple paths in the dual graph and hence we can keep only the lowest-cost edge of a set of parallel edges.
For the same reason, we can ignore any self-loops in the dual graph and only consider simple weighted
graphs.
We refer to the text books by Agnarsson and Greenlaw~\cite{agnarsson2006graphtheory} and by Diestel~\cite{diestel2016graphtheory} for an introduction to general graph theory, planar graphs, and other graph theoretic notation.

We refer to a planar graph with a specific embedding as a \emph{plane graph}, and whenever we talk about \emph{dual graphs}, we mean duals
of plane graphs.
The dual is defined as usual, and for a plane graph
$G$, we refer to the dual graph as $G^\star$.  Since every edge
$e \in E(G)$ corresponds to a dual edge in $G^\star$, we will refer to
the dual edge of $e$ as~$e^\star$.  Similarly, if $S \subseteq E$, then
we write $S^\star$ for~$\{e^\star \mid e \in S\}$. For weighted graphs, we let $w(e^\star) = w(e)$.

A set of edges $C$ is called a \emph{cut} if there exists a set~$S \subseteq V(G)$ such that $C$ contains exactly those edges with exactly one endpoint in $S$.
A cut $C = \cut(S)$ is an $s$-$t$-cut if $s \in S$ and $t \notin S$ or vice versa.
A cut $C$ is \emph{minimal} if no proper subset of $C$ is a cut.
The \emph{weight} of a cut (or a path/cycle) is the sum of the weights of its edges.
In this work, we only consider \emph{minimal non-empty cuts}.

\subsection{Derigs' Shortest Odd Path}
\label{sec:derigs}

The main subroutine of our algorithm is an algorithm that finds a
shortest odd $s$-$t$-path.  Here, \emph{shortest} means lowest cost, and
\emph{odd} means that we need an odd number of edges.  There are several
algorithms solving this problem, and in graphs without weights, there is
an $O(m)$ algorithm due to Lapaugh and
Papadimitriou~\cite{lapaugh1984evenpathproblem}\footnote{This algorithm
  finds a shortest even path, but finding an odd path is simple by
  adding a pendant vertex $s'$ to $s$ and search for an even
  $s'$-$t$-path instead.}.  We note in passing that on directed graphs,
checking the existence of an odd path is
NP-complete~\cite{lapaugh1984evenpathproblem}.

For the weighted version, we use Derigs' algorithm for shortest odd
path~\cite{derigs1985efficientdijkstralike}, which is a Dijkstra-like
algorithm, but with an additional blossom step.  Both Derigs' and
Lapaugh and Padadimitriou's algorithms rely on finding augmenting paths
in a transformed graph.

\begin{theorem}[Derigs' algorithm~\cite{derigs1985efficientdijkstralike}]
    \label{thm:derig}
  Given a weighted graph~$G$ and two vertices~$s$ and~$t$, we can in
  $O(m \log n)$ time find a cheapest \emph{odd} $s$-$t$-path, or
  correctly conclude that none exist.
\end{theorem}

\noindent
The main idea behind Derigs' algorithm is the following.  Let $G$ be an input graph, with two designated vertices~$s$ and~$t$.  We want to find an \emph{even} $s$-$t$-path.  We create a new  graph~$G'$ consisting of two copies of~$G$, and we let $v'$ denote the copy of $v$, and refer to $v'$ as the mirror vertex of $v$.  Now, add all edges $v,v'$ to $G'$ and finally delete $s'$ and $t$ from $G'$.
Create a matching $M = \{(v,v') \mid v \in V(G)\}$.  Note that neither~$s$ nor~$t'$ are in any edges of $M$.

Given a graph and a matching, an \emph{alternating path} is a path with an odd number of edges in the graph in which every edge of even index is in the matching and the other edges are not in the matching.  It can be shown that $G$ has an even $s$-$t$-path if and only if $G'$ and $M$ has an alternating path.

\subsection{An Application of Shortest Odd Path}
\label{sec:anapplication}

Before explaining the algorithm for \pname{Network Diversion}, which
builds on shortest odd path, we give another useful application of
shortest odd path, that we call \pname{Detour Path}. This problem is defined as follows.

\begin{problembox}{Detour Path}
  \Input & An edge-weighted graph $G$, two vertices $s$ and $t$, and an edge~$b$.\\
  \Prob  & Among all $s$-$t$-path that use~$b$, find a shortest one (or conclude that no such path exists).
\end{problembox}

\noindent
In \pname{Detour Path}, we are given an edge-weighted graph~$G$, two vertices~$s$
and~$t$, and an edge~$b$, and we are asked to find a shortest path from~$s$ to~$t$ that uses~$b$.  Since we are asking for a path rather
than a walk, we need to find two disjoint paths---one from~$s$ to one of
the endpoints of~$b$ and one from~$t$ to the other endpoint of~$b$.
This can be solved by running two instances of \pname{Min-Sum Two Disjoint Paths}.
However, we can also solve the problem faster by using shortest odd path as follows.
On input~$(G,s,t,b)$, create the graph~$G'$ by
subdividing every edge of~$G$ except~$b$.
If $e_1, e_2$ are the subdivision edges that replace~$e$, let
$w(e_1) = w(e_2) = \frac{1}{2}w(e)$.
See \Cref{figure:subdividing-detours} for an illustration.
Now, it is clear that in~$G'$, any
path between vertices in the original graph~$G$ is odd if and only if the path contains~$b$.
Moreover, the length of paths are preserved.
\begin{figure}[t]
  \centering
  \begin{subfigure}{.43\textwidth}
    \centering
    \scalebox{0.5}{
      \begin{tikzpicture}
        \tikzstyle{every node}=[circle, fill=lightgray, draw=black, inner sep=2pt, minimum size=1.5em, font=\footnotesize, text=black]
        \tikzstyle{edge}=[gray, line width=1.5mm]

        \node (s) at (0,0) {$s$};
        \node (a) at (1.5,1.5) {};
        \node (b) at (3,0) {};
        \node (c) at (1.5,-1.5) {};
        \node (d) at (4.5,-1.5) {};
        \node (e) at (6,0) {};
        \node (f) at (4.5,1.5) {};
        \node (t) at (7.5,1.5) {$t$};

        \draw[edge] (s) -- (a) -- (b) -- (e);
        \draw[edge] (c) -- (b) -- (d) -- (e) -- (f) -- (t);
        \draw[edge] (a) -- (f);

        \tikzstyle{edge}=[myred, line width=1.5mm]
        \draw[edge] (c) -- (d);
      \end{tikzpicture}
    }
    \caption{An instance of \textsc{Shortest Detour Path} with the detour edge marked in red.}
    \label{figure:detour}
  \end{subfigure}\hfill \begin{subfigure}{.43\textwidth}
    \centering
    \scalebox{0.5}{
      \begin{tikzpicture}
        \tikzstyle{every node}=[circle, fill=lightgray, draw=black, inner sep=2pt, minimum size=1.5em, font=\footnotesize, text=black]
        \tikzstyle{edge}=[gray, line width=1.5mm]

        \node (s) at (0,0) {$s$};
        \node (a) at (1.5,1.5) {};
        \node (b) at (3,0) {};
        \node (c) at (1.5,-1.5) {};
        \node (d) at (4.5,-1.5) {};
        \node (e) at (6,0) {};
        \node (f) at (4.5,1.5) {};
        \node (t) at (7.5,1.5) {$t$};

        \draw[edge] (s) -- (a) -- (b) -- (e);
        \draw[edge] (c) -- (b) -- (d) -- (e) -- (f) -- (t);
        \draw[edge] (a) -- (f);

        \tikzstyle{edge}=[myred, line width=1.5mm]
        \draw[edge] (c) -- (d);

        \tikzstyle{every node}=[circle, fill=lightgray, draw=black, inner sep=2pt, minimum size=1em, font=\footnotesize, text=black]
        \node (sa) at (0.75,0.75) {};
        \node (af) at (3,1.5) {};
        \node (ab) at (2.25,0.75) {};
        \node (bc) at (2.25,-0.75) {};
        \node (bd) at (3.75,-0.75) {};
        \node (be) at (4.5,0) {};
        \node (de) at (5.25,-0.75) {};
        \node (ef) at (5.25,0.75) {};
        \node (ft) at (6,1.5) {};
      \end{tikzpicture}
    }
    \caption{Instance of \textsc{Shortest Odd Path}:  All edges but the detour edge are subdivided.}
    \label{figure:subdivided-detour}
  \end{subfigure}\caption{\textsc{Shortest Detour Path} reduced to \textsc{Shortest Odd Path} by subdividing all edges except the detour.}
  \label{figure:subdividing-detours}
\end{figure}
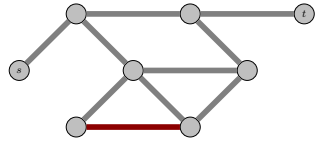
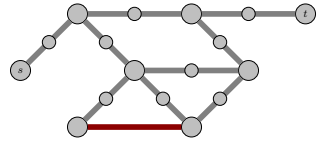

\section{Network Diversion on Planar Graphs}
\label{sec:network-diversion}
In this section, we present our algorithm for \textsc{Network Diversion} on planar input graphs.
We start with a collection of relevant results from the known literature.

First, we use the following strong connection between cuts in planar graphs and cycles in their dual graph.
\begin{theorem}[{\cite[Proposition 4.6.1]{diestel2016graphtheory}}]
  Minimal cuts in the graph corresponds to simple cycles in the dual graph.
\end{theorem}

\noindent
Second, it was previously observed (for example by Bentert et al.~\cite{bentert2024twosets}) that a minimal cut separates two vertices if and only if the corresponding cycle crosses a path between those two vertices an odd number of times.
\begin{observation}
  Let $P$ be an $s$-$t$-path in $G$.  A cycle~$C^\star$ in the dual
  graph corresponds to a minimal $s$-$t$-cut if and only if it crosses~$P$ an odd number of times.  Moreover, if $C^\star$ is a minimum-cost
  cycle satisfying the above criteria, then~$C$ is a minimum-weight minimal $s$-$t$-cut.
\end{observation}

\noindent
We use this observation in the following reformulation.

\begin{corollary}
    \label{cor:oddcrossing}
  Let $E_P$ be the edges of any $s$-$t$-path, and $E_P^\star$ the
  corresponding dual edges.  Any cycle $C^\star$ in the dual graph
  corresponds to an $s$-$t$-cut in $G$ if and only if $C^\star$ contains
  an odd number of edges from $E_P^\star$.
\end{corollary}

\noindent
Finally, using Derigs' algorithm (\Cref{thm:derig}), we obtain the following result as a corollary.

\begin{corollary}
  \label{cor:path-odd-number-f}
  Let~$G$ be a graph,~$F \subseteq E(G)$ be a subset of edges of~$G$,
  and~$s$ and~$t$ be two vertices.  We can find in $O(m \log n)$ time an
  $s$-$t$-path that uses an odd number of edges from~$F$, or correctly
  conclude that none exist.
\end{corollary}
\begin{proof}
  Let~$G'$ be the graph that results from subdividing every edge in~$G$
  that does not appear in~$F$.
If~$e_1, e_2$ are the subdivision edges that replace~$e$, let
  $w_1 = \lceil w/2 \rceil$ and $w_2 = \lfloor w/2 \rfloor$.
Now, the length of any path must be congruent (mod 2) to the number of
  edges it contains in~$F$.
\end{proof}

\noindent
Using the known fact that planar graphs are sparse, we obtain the following.

\begin{lemma}[Folklore]
  \label{lem:m-is-n}
  For simple planar graphs, the number~$m$ of edges is upper-bounded by~${3n -6}$,
  where $n$ is the number of vertices.  Hence, $m \in O(n)$ and the
  algorithms above run in~$O(n \log n)$~time.
\end{lemma}

\subsection{Algorithm for Planar Network Diversion}

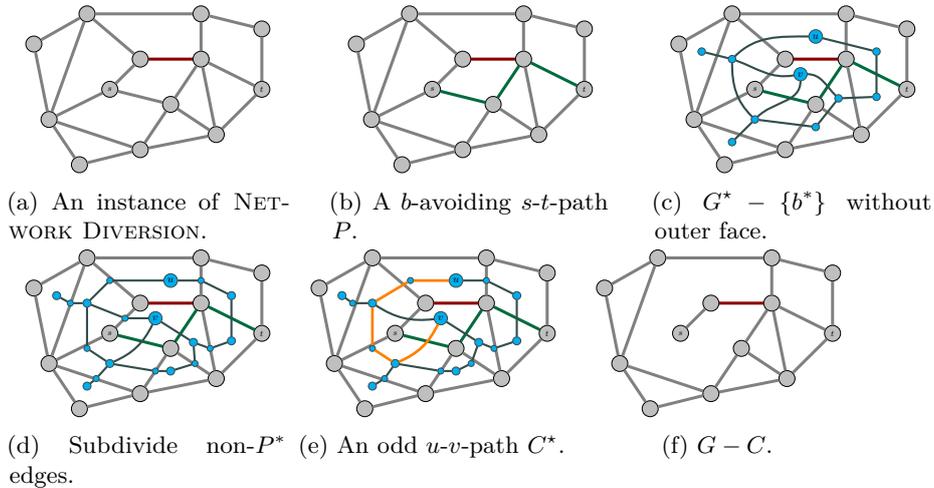
\begin{figure}[t]
  \begin{subfigure}[t]{.3\textwidth}
    \centering
    \scalebox{0.4}{
      \begin{tikzpicture}
        \tikzstyle{every node}=[circle, fill=lightgray, draw=black, inner sep=2pt, minimum size=1.5em, font=\footnotesize, text=black]
        \tikzstyle{edge}=[gray, line width=1mm]

        \node (s) at (0,0) {$s$};
        \node (a) at (2,-0.5) {};
        \node (b) at (1,-2) {};
        \node (c) at (1,1) {};
        \node (d) at (3.5,-1.5) {};
        \node (e) at (3,1) {};
        \node (f) at (-2,-1) {};
        \node (t) at (5,0) {$t$};
        \node (g) at (-1,-2.5) {};
        \node (h) at (3,2.5) {};
        \node (i) at (-0.75,2.5) {};
        \node (j) at (5,2) {};
        \node (k) at (-2.5,1.5) {};

        \draw[edge] (t) -- (j) -- (h) -- (i) -- (k) -- (f) -- (g) -- (b) -- (d) -- (e) -- (h);
        \draw[edge] (s) -- (c) -- (i) -- (f) -- (s);
        \draw[edge] (f) -- (b) -- (a) -- (d) -- (t);
        \draw[edge] (s) -- (a) -- (e) -> (t);

        \tikzstyle{edge}=[myred, line width=1mm]
        \draw[edge] (c) -- (e);
      \end{tikzpicture}
    }
    \caption{An instance of \textsc{Network Diversion}.}
  \end{subfigure}
\begin{subfigure}[t]{.3\textwidth}
    \centering
    \scalebox{0.4}{

      \begin{tikzpicture}
        \tikzstyle{every node}=[circle, fill=lightgray, draw=black, inner sep=2pt, minimum size=1.5em, font=\footnotesize, text=black]
        \tikzstyle{edge}=[gray, line width=1mm]

        \node (s) at (0,0) {$s$};
        \node (a) at (2,-0.5) {};
        \node (b) at (1,-2) {};
        \node (c) at (1,1) {};
        \node (d) at (3.5,-1.5) {};
        \node (e) at (3,1) {};
        \node (f) at (-2,-1) {};
        \node (t) at (5,0) {$t$};
        \node (g) at (-1,-2.5) {};
        \node (h) at (3,2.5) {};
        \node (i) at (-0.75,2.5) {};
        \node (j) at (5,2) {};
        \node (k) at (-2.5,1.5) {};

        \draw[edge] (t) -- (j) -- (h) -- (i) -- (k) -- (f) -- (g) -- (b) -- (d) -- (e) -- (h);
        \draw[edge] (s) -- (c) -- (i) -- (f) -- (s);
        \draw[edge] (f) -- (b) -- (a) -- (d) -- (t);

        \tikzstyle{edge}=[myred, line width=1mm]
        \draw[edge] (c) -- (e);

        \tikzstyle{edge}=[mygreen, line width=1mm]
        \draw[edge] (s) -- (a) -- (e) -> (t);
      \end{tikzpicture}
    }
    \caption{A $b$-avoiding $s$-$t$-path $P$.}
  \end{subfigure}
\begin{subfigure}[t]{.3\textwidth}
    \centering
    \scalebox{0.4}{
      \begin{tikzpicture}
        \tikzstyle{every node}=[circle, fill=lightgray, draw=black, inner sep=2pt, minimum size=1.5em, font=\footnotesize, text=black]
        \tikzstyle{edge}=[gray, line width=1mm]

        \node (s) at (0,0) {$s$};
        \node (a) at (2,-0.5) {};
        \node (b) at (1,-2) {};
        \node (c) at (1,1) {};
        \node (d) at (3.5,-1.5) {};
        \node (e) at (3,1) {};
        \node (f) at (-2,-1) {};
        \node (t) at (5,0) {$t$};
        \node (g) at (-1,-2.5) {};
        \node (h) at (3,2.5) {};
        \node (i) at (-0.75,2.5) {};
        \node (j) at (5,2) {};
        \node (k) at (-2.5,1.5) {};

        \draw[edge] (t) -- (j) -- (h) -- (i) -- (k) -- (f) -- (g) -- (b) -- (d) -- (e) -- (h);
        \draw[edge] (s) -- (c) -- (i) -- (f) -- (s);
        \draw[edge] (f) -- (b) -- (a) -- (d) -- (t);

        \tikzstyle{edge}=[myred, line width=1mm]
        \draw[edge] (c) -- (e);

        \tikzstyle{edge}=[mygreen, line width=1mm]
        \draw[edge] (s) -- (a) -- (e) -> (t);

        \tikzstyle{every node}=[circle, draw=myteal, fill=cyan, inner sep=2pt, minimum size=0.7em]
        \node (aces) at (1.5,0.5) {$v$};
        \node (cehi) at (2,1.75) {$u$};
        \node (ehjt) at (4,1.25) {};
        \node (det) at (4,-0.25) {};
        \node (ade) at (2.75,-0.3) {};
        \node (abd) at (2,-1.25) {};
        \node (abfs) at (0,-1) {};
        \node (bfg) at (-0.75,-1.75) {};
        \node (cifs) at (-0.75,1) {};
        \node (fik) at (-1.75,1.25) {};

        \tikzstyle{edge}=[myteal, line width=0.7mm]
        \draw[edge] (fik) to (cifs) to[out=45,in=180] (cehi) -- (ehjt) -- (det) -- (ade) -- (abd);
        \draw[edge] (abd) to (abfs);
        \draw[edge] (abfs) to (bfg);
        \draw[edge] (aces) to[out=0,in=135] (ade);
        \draw[edge] (aces) to[out=190,in=-20] (cifs);
        \draw[edge] (cifs) to[out=-90,in=135] (abfs);
        \draw[edge] (abfs) to[out=45,in=270] (aces);
      \end{tikzpicture}
    }
    \caption{$G^\star - \{b^*\}$  without outer face.}
  \end{subfigure}
\begin{subfigure}[t]{.3\textwidth}
    \centering
    \scalebox{0.4}{
      \begin{tikzpicture}
        \tikzstyle{every node}=[circle, fill=lightgray, draw=black, inner sep=2pt, minimum size=1.5em, font=\footnotesize, text=black]
        \tikzstyle{edge}=[gray, line width=1mm]

        \node (s) at (0,0) {$s$};
        \node (a) at (2,-0.5) {};
        \node (b) at (1,-2) {};
        \node (c) at (1,1) {};
        \node (d) at (3.5,-1.5) {};
        \node (e) at (3,1) {};
        \node (f) at (-2,-1) {};
        \node (t) at (5,0) {$t$};
        \node (g) at (-1,-2.5) {};
        \node (h) at (3,2.5) {};
        \node (i) at (-0.75,2.5) {};
        \node (j) at (5,2) {};
        \node (k) at (-2.5,1.5) {};

        \draw[edge] (t) -- (j) -- (h) -- (i) -- (k) -- (f) -- (g) -- (b) -- (d) -- (e) -- (h);
        \draw[edge] (s) -- (c) -- (i) -- (f) -- (s);
        \draw[edge] (f) -- (b) -- (a) -- (d) -- (t);

        \tikzstyle{edge}=[myred, line width=1mm]
        \draw[edge] (c) -- (e);

        \tikzstyle{edge}=[mygreen, line width=1mm]
        \draw[edge] (s) -- (a) -- (e) -> (t);

        \tikzstyle{every node}=[circle, draw=myteal, fill=cyan, inner sep=2pt, minimum size=0.75em]
        \node (aces) at (1.5,0.5) {$v$};
        \node (cehi) at (2,1.75) {$u$};
        \node (ehjt) at (4,1.25) {};
        \node (det) at (4,-0.25) {};
        \node (ade) at (2.75,-0.3) {};
        \node (abd) at (2,-1.25) {};
        \node (abfs) at (0,-1) {};
        \node (bfg) at (-0.75,-1.75) {};
        \node (cifs) at (-0.75,1) {};
        \node (fik) at (-1.75,1.25) {};

        \tikzstyle{every node}=[circle, draw=myteal, fill=cyan, inner sep=2pt, minimum size=0.35em]
        \node (fi) at (-1.3,1) {};
        \node (fs) at (-0.75,-0.5) {};
        \node (bf) at (-0.45,-1.5) {};
        \node (ci) at (0.0,1.75) {};
        \node (eh) at (3,1.75) {};
        \node (de) at (3.3,-0.5) {};
        \node (ad) at (2.8,-1) {};
        \node (ab) at (1.5,-1.25) {};
        \node (sc) at (0.5,0.5) {};

        \tikzstyle{edge}=[myteal, line width=0.7mm]
        \draw[edge] (fik) -- (fi) -- (cifs) -- (ci) -- (cehi) -- (eh) -- (ehjt) -- (det) -- (de) -- (ade) -- (ad) -- (abd) -- (ab) -- (abfs) -- (bf) -- (bfg);
        \draw[edge] (ade) -- (aces) -- (sc) -- (cifs) -- (fs) -- (abfs) to [bend right=20] (aces);
      \end{tikzpicture}
    }
    \caption{Subdivide non-$P^*$ edges.}
  \end{subfigure}
\begin{subfigure}[t]{.3\textwidth}
    \centering
    \scalebox{0.4}{
      \begin{tikzpicture}
        \tikzstyle{every node}=[circle, fill=lightgray, draw=black, inner sep=2pt, minimum size=1.5em, font=\footnotesize, text=black]
        \tikzstyle{edge}=[gray, line width=1mm]

        \node (s) at (0,0) {$s$};
        \node (a) at (2,-0.5) {};
        \node (b) at (1,-2) {};
        \node (c) at (1,1) {};
        \node (d) at (3.5,-1.5) {};
        \node (e) at (3,1) {};
        \node (f) at (-2,-1) {};
        \node (t) at (5,0) {$t$};
        \node (g) at (-1,-2.5) {};
        \node (h) at (3,2.5) {};
        \node (i) at (-0.75,2.5) {};
        \node (j) at (5,2) {};
        \node (k) at (-2.5,1.5) {};

        \draw[edge] (t) -- (j) -- (h) -- (i) -- (k) -- (f) -- (g) -- (b) -- (d) -- (e) -- (h);
        \draw[edge] (s) -- (c) -- (i) -- (f) -- (s);
        \draw[edge] (f) -- (b) -- (a) -- (d) -- (t);

        \tikzstyle{edge}=[myred, line width=1mm]
        \draw[edge] (c) -- (e);

        \tikzstyle{edge}=[mygreen, line width=1mm]
        \draw[edge] (s) -- (a) -- (e) -> (t);

        \tikzstyle{every node}=[circle, draw=myteal, fill=cyan, inner sep=2pt, minimum size=0.75em]
        \node (aces) at (1.5,0.5) {$v$};
        \node (cehi) at (2,1.75) {$u$};
        \node (ehjt) at (4,1.25) {};
        \node (det) at (4,-0.25) {};
        \node (ade) at (2.75,-0.3) {};
        \node (abd) at (2,-1.25) {};
        \node (abfs) at (0,-1) {};
        \node (bfg) at (-0.75,-1.75) {};
        \node (cifs) at (-0.75,1) {};
        \node (fik) at (-1.75,1.25) {};

        \tikzstyle{every node}=[circle, draw=myteal, fill=cyan, inner sep=2pt, minimum size=0.35em]
        \node (fi) at (-1.4,1) {};
        \node (fs) at (-0.75,-0.5) {};
        \node (bf) at (-0.35,-1.5) {};
        \node (ci) at (0.5,1.75) {};
        \node (eh) at (3,1.75) {};
        \node (de) at (3.25,-0.5) {};
        \node (ad) at (2.5,-1) {};
        \node (ab) at (1.5,-1.25) {};

        \tikzstyle{edge}=[myteal, line width=0.7mm]
        \draw[edge] (fik) -- (fi) -- (cifs);
        \draw[edge] (cehi) -- (eh) -- (ehjt) -- (det) -- (de) -- (ade) -- (ad) -- (abd) -- (ab) -- (abfs) -- (bf) -- (bfg);
        \draw[edge] (ade) -- (aces) to [bend left=20] (cifs);

        \tikzstyle{edge}=[orange, line width=0.9mm]
        \draw[edge] (cehi) -- (ci) -- (cifs) -- (fs) -- (abfs) to [bend right=20] (aces);
\end{tikzpicture}
    }
    \caption{An odd $u$-$v$-path $C^\star$.}
  \end{subfigure}
\begin{subfigure}[t]{.3\textwidth}
    \centering
    \scalebox{0.4}{
      \begin{tikzpicture}
        \tikzstyle{every node}=[circle, fill=lightgray, draw=black, inner sep=2pt, minimum size=1.5em, font=\footnotesize, text=black]
        \tikzstyle{edge}=[gray, line width=1mm]

        \node (s) at (0,0) {$s$};
        \node (a) at (2,-0.5) {};
        \node (b) at (1,-2) {};
        \node (c) at (1,1) {};
        \node (d) at (3.5,-1.5) {};
        \node (e) at (3,1) {};
        \node (f) at (-2,-1) {};
        \node (t) at (5,0) {$t$};
        \node (g) at (-1,-2.5) {};
        \node (h) at (3,2.5) {};
        \node (i) at (-0.75,2.5) {};
        \node (j) at (5,2) {};
        \node (k) at (-2.5,1.5) {};

        \draw[edge] (t) -- (j) -- (h) -- (i) -- (k) -- (f) -- (g) -- (b) -- (d) -- (e) -- (h);
        \draw[edge] (c) -- (s);
        \draw[edge] (a) -- (e) -- (t);
        \draw[edge] (i) -- (f) -- (b) -- (a) -- (d) -- (t);

        \tikzstyle{edge}=[myred, line width=1mm]
        \draw[edge] (c) -- (e);
      \end{tikzpicture}
    }
    \caption{$G - C$.}
  \end{subfigure}
  \caption{An illustration of the different steps in our algorithm for \pname{Network Diversion} on planar graphs. We assume here that all edges have the same weight. Steps~(a) -- (c) are directly from our algorithm and Step~(d) shows the reduction from finding a shortest path that crosses the computed path an odd number of times to finding a shortest odd path given by \Cref{cor:path-odd-number-f}.
  Step~(e) shows the solution of running Derigs' algorithm and Step~(f) shows the solution to the input instance we computed.}
  \label{fig:nd-solved}
\end{figure}

Our algorithm for \pname{Network Diversion} on planar graphs can be explained
relatively succinctly, and is illustrated in \Cref{fig:nd-solved}.
The algorithm proceeds in five steps described below, and the running time
follows from the previous results (\Cref{cor:path-odd-number-f} and \Cref{lem:m-is-n}).

\bigskip
\noindent
\textbf{Computing the dual.}  Computing the dual graph of a given
plane graph is folklore wisdom, but we include a high-level description
for completeness.
The input to the algorithm is a plane graph, two vertices and an edge.
In our implementation, all input graphs are given as straight-line
embeddings (which always exist).
However, we only use the fact that every vertex has an ordered list of edges incident to it (ordered in counter-clockwise fashion).
Computing the dual of a plane graph $G$ then simply proceeds as follows.  Let $e = uv$ be an edge, and let $f_{uv}$ and $f_{vu}$ be the faces one each side of
$e$.  We find these faces by traversing directed edges, starting with
$uv$, and following to the \emph{next} edge incident on $v$ after $e$.
When we return to $uv$, we have computed $f_{uv}$, and we can traverse
the edge the other direction, $vu$ to obtain $f_{vu}$.

This creates a map from each edge to its two (or one) neighboring faces,
and we thus have the vertices and the edges of the dual graph.
In total, we visit each edge exactly twice (once in each orientation), and thus it takes linear $O(n)$ time to compute the dual.
We note that \emph{bridges} are only incident on one face.

We can finally describe our algorithm.
To this end, let~$G=(V,E)$ be an edge-weighted graph and let two vertices~$s$ and~$t$ and an edge~$b$ be given.
Let~$b^\star=\{u,v\}$ be the dual of~$b$.
We assume without loss of generality that~$G$ is connected.

\begin{enumerate}
\item Compute a path~$P$ between~$s$ and~$t$ in $G-\{b\}$ or conclude that~$b$ is already a bridge between~$s$ and~$t$ in~$G$.
\item Compute the dual graph~$G^\star$ of~$G$.
\item Find a shortest $u$-$v$-path $C^\star$ in $G^\star - \{b^\star\}$ that uses an odd number of edges from $E_P^\star$ using Derigs' algorithm, or conclude that none exist.
\item Return $C$.
\end{enumerate}

\medskip
\noindent
Step 3 is where we overcome the issue raised by Cullenbine et al.~\cite{cullenbine2013theoreticalcomputational}.  We compute a shortest  odd path guaranteed to be \emph{simple} in $G^\star$, which necessarily corresponds to a minimal cut in~$G$.

\begin{theorem}
    Given a plane graph $G$, $s$, $t$, and $b$, we can compute the minimum diversion set or correctly conclude that none exist in $O(n \log n)$ time.
\end{theorem}
\begin{proof}
    If~$s$ and~$t$ do not belong to the same connected component of~$G$, then we can safely output no.
    If there is no~$s$-$t$-path in~$G - \{b\}$, then~$b$ is already an~$s$-$t$-bridge and we can safely output yes.
Otherwise, we find an $s$-$t$-path~$P$ that does not use $b$.
    What remains is to show that a minimum minimal $s$-$t$-cut containing $b$ corresponds to a shortest cycle using $b^\star$ that uses an odd number of edges from $P^\star$.
    The cost aspect follows from the correspondence between length of cycles in the dual and cost of minimal cuts in the primal graph, so the only thing remaining is to show that this cycle is indeed an $s$-$t$-cut in the primal.
By \Cref{cor:oddcrossing}, given a cut~$C$, and a $u$-$v$-path $S$, $u$ and~$v$ are on different sides of the cut if and only if~$S$ uses an odd number of edges from~$C$, or equivalently, $C$ uses an odd number of edges from~$S$.
    Hence, a shortest cycle containing $b^\star$ and using an odd number of edges from $P^\star$ indeed corresponds to an $s$-$t$-cut in the primal.
It follows that the algorithm outlined above is correct.
For the running time, note that computing the dual $G^\star$, finding any $s$-$t$-path, and subdividing edges, all take $O(n+m)=O(n)$ time. Running Derigs' algorithm takes $O(n \log n)$ time by \Cref{lem:m-is-n}.
\end{proof}

\section{Implementation and Experimental Setup}
\label{sec:implementation}

\begin{table}
  \small
  \centering
    \begin{tabular}{l  r  r  r  r}
      & n & m & na\"ive & union--find \\
      \cmidrule{2-5}
       Power         & 1723   & 2394   & 2 ms   & 2 ms   \\
       Oldenburg     & 6106   & 7035   & 6 ms   & 5 ms   \\
       San~Joaquin   & 18263  & 23874  & 22 ms  & 18 ms  \\
       Cali~Road     & 21048  & 21693  & 21 ms  & 18 ms  \\
       Musae~Git     & 37700  & 289003 & 125 ms & 123 ms \\
       SF~Road~N     & 174956 & 223001 & 225 ms & 191 ms \\
       Ca~Citeseer   & 227321 & 814137 & 629 ms & 608 ms \\
      \cmidrule{2-5}
       Delaunay~50k  & 50000  & 149961 & 106 ms & 96 ms  \\
       Delaunay~100k & 100000 & 299959 & 219 ms & 205 ms \\
       Delaunay~150k & 150000 & 449965 & 346 ms & 315 ms \\
       Delaunay~200k & 200000 & 599961 & 476 ms & 449 ms \\
    \end{tabular}
    \caption{Comparison of our two implementations for computing a shortest odd path on different graphs. The first column shows the name of the data set (the first seven datasets are real world data, the next bottom four are randomly generated Delaunay graphs of different sizes). The first columns~$n$ and~$m$ show the number of vertices and edges in the respective graph and the last two columns denote the running times of the na\"ive implemenation and the implementation using the union--find data structure.}
    \label{tab:bench-shortest}
\end{table}

To the best of our knowledge, we provide the first implementation of Derigs' shortest odd path algorithm.
We believe that this implementation might be of independent interest.
It can handle edge-weighted graphs, and solves graphs of one million edges in less than a second.

In addition, we implement the algorithm for \textsc{Network Diversion} on planar input graphs.
We note here that in our specific implementation, we only work with straight-line embeddings of planar graphs, but this can easily be generalized and replaced with any library that can, given a
plane graph, return a mapping between vertices and faces in the input graph to faces and vertices, respectively, of the dual graph.
Given a plane graph, computing the dual, and finding a $b$-avoiding
$s$-$t$-path are both simple $O(n)$ algorithms, and the bulk of the
running time of the algorithm is finding a shortest odd path using Derigs'
algorithm.

\label{subsec:speed}
We want to point out one particular detail about our implementation.
We use Edmond's Blossom
algorithm~\cite{edmonds1965pathstrees} for computing matchings (or, strictly speaking, \emph{alternating paths}).
For the following technical discussion, we assume familiarity with this algorithm.
When we have found and computed a blossom, we shrink it into a pseudonode by setting the base of all its vertices to the base of the blossom.  Whenever we consider a potential blossom edge, we see if the two vertices have the same base, and if so, then they are in fact already in the same pseudonode and the edge can be discarded.
Whenever we set~$u$ to have the base~$\beta$, we also have to see if any other vertices have~$u$ as their base and set their bases to~$\beta$ as well.
Derigs does not specify which data structure to use to update these bases efficiently and the na\"ive way would search through all vertices in the graph in linear time,
which would not actually give the claimed running time.
Now we go through only the vertices that have~$u$ as their base, in time
linear to the count of vertices that need to be updated.
The second version is to use a structure resembling union--find, where
each disjoint set and its representative is a blossom and its base.  To
update the base of~$u$ we simply set the new base and do nothing else.
When we require the base of a vertex, we recursively query its
representative's base and contract the path along the way in the style
of union--find.

\begin{figure}[t]
  \centering
  \includegraphics[width=0.75\textwidth]{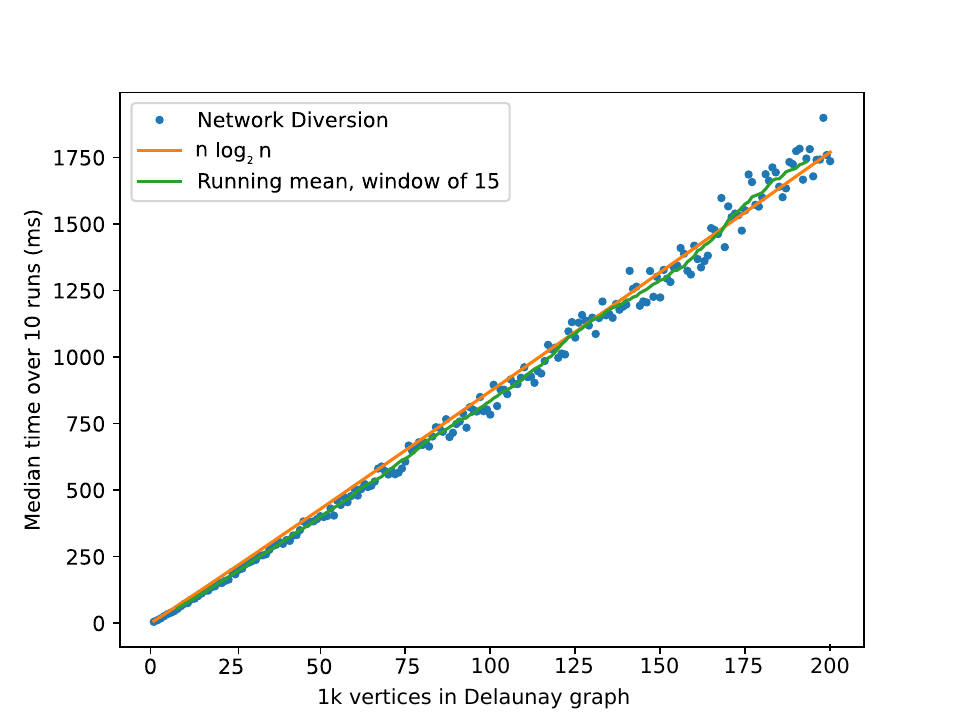}
  \caption{Running times of our implementation for \textsc{Network Diversion} on Delaunay graphs of different sizes. The x-axis shows the number of vertices in 1000s and the blue dots show the running time in milliseconds. The green line shows a running mean running time over the previous 15 instances. The orange line shows the line~$\nicefrac{n \log_2 n}{100}$ ms.}
  \label{fig:nd-bench-delauay}
\end{figure}

All of our experiments were performed on a laptop with an Intel(R) Core(TM) i5-8400 CPU @ 2.80GHz with 192KiB, 1.5 MiB, and 9 MiB L1, L2, and L3 cache, respectively, and 15G RAM, running on GNU/Linux Ubuntu 20.04.6 LTS.
We use Rust edition 2021 compiled with \texttt{rustc 1.81.0-nightly}.

We tested our implementation of Derigs' algorithm for finding a shortest odd path on both real-world (non-planar) data and on randomly generated Delaunay graphs.
A Delaunay graph is the graph formed by the edges of the Delaunay triangulation of a set of randomly drawn points in the plane, where vertices represent the points and edges connect pairs of points that form the sides of the triangles in the triangulation.
Delaunay graphs are planar and triangulated, meaning that every face (except the outer) is a triangle.

For \textsc{Network Diversion} on planar graphs, we implemented and compared our algorihm to the previous implementation by Cullenbine et al.~\cite{cullenbine2013theoreticalcomputational}.
They used both real-world and synthetic data, where they generated weighted grid graphs.
Unfortunately, we do not have access to the real-world data sets they used or the exact MILP formulation they used.
Instead, we generated random planar graphs similar to the kind of graphs that Cullenbine et al.\ generated.
In particular, we generated randomly weighted grid graphs of different sizes and chose~$s, t$, and~$b$ uniformly at random.
We also tested our implementation on the generated Delaunay graphs.
We ran each experiment 100 times and took the average running time.
Since our algorithm does not in any way use the fact that the given graph is a grid graph or a Delaunay graph and the running times for similar-sized graphs are comparable, it is safe to assume that the running times scale with the number of vertices, faces, and edges.

\section{Experimental Results}
\label{sec:experiments}

In this section, we discuss our experimental findings.
First, we tested our implementation for finding a \emph{shortest odd path} both with a na\"ive implementation of Edmond's Blossom
algorithm and with the union--find data structure discussed in \Cref{subsec:speed}.
The results are depicted in \Cref{tab:bench-shortest} and the implementation with union--find turned out to be faster by roughly 10\% on average.
It sped up the algorithm especially on sparser graphs but we also mention that it occasionally was slower when the number of edges approached~$\binom{n}{2}$.
However, since planar graphs are sparse, we have chosen union--find for the remainder of the experiments.

\begin{table}
  \small
  \centering
    \begin{tabular}{l  r  r  r  r r r}
      & n & m & $t_{25}$ & $t_{50}$ & $t_{75}$ \\
      \cmidrule{2-6}
Grid $10 \times 10$&\num{100}&\num{180}&0.00&0.00&0.00\\
Grid $20 \times 20$&\num{400}&\num{760}&0.00&0.00&0.00\\
Grid $30 \times 30$&\num{900}&\num{1740}&0.00&0.00&0.00\\
Grid $40 \times 40$&\num{1600}&\num{3120}&0.00&0.00&0.00\\
Grid $50 \times 50$&\num{2500}&\num{4900}&0.01&0.01&0.01\\
Grid $60 \times 60$&\num{3600}&\num{7080}&0.01&0.01&0.01\\
Grid $70 \times 70$&\num{4900}&\num{9660}&0.01&0.02&0.02\\
Grid $80 \times 80$&\num{6400}&\num{12640}&0.02&0.02&0.02\\
Grid $90 \times 90$&\num{8100}&\num{16020}&0.03&0.03&0.03\\
Grid $100 \times 100$&\num{10000}&\num{19800}&0.04&0.04&0.04\\
Grid $200 \times 200$&\num{40000}&\num{79600}&0.16&0.16&0.17\\
Grid $300 \times 300$&\num{90000}&\num{179400}&0.35&0.38&0.40\\
Grid $400 \times 400$&\num{160000}&\num{319200}&0.66&0.74&0.77\\
Grid $500 \times 500$&\num{250000}&\num{499000}&0.90&1.23&1.28\\
Grid $600 \times 600$&\num{360000}&\num{718800}&1.62&1.84&1.99\\
Grid $700 \times 700$&\num{490000}&\num{978600}&2.42&2.55&2.69\\
Grid $800 \times 800$&\num{640000}&\num{1278400}&3.19&3.50&3.67\\
Grid $900 \times 900$&\num{810000}&\num{1618200}&4.06&4.71&4.92\\
Grid $1000 \times 1000$&\num{1000000}&\num{1998000}&5.58&5.81&6.06\\
Grid $2000 \times 2000$&\num{4000000}&\num{7996000}&25.96&28.88&29.74\\
    \end{tabular}
    \caption{Running times in seconds for solving \textsc{Network Diversion} on weighted grid graphs.
    The weights on the edges are floats chosen uniformly at random between~0 and 1000. The times are the 25\%, 50\%, and 75\% quartiles, respectively.}
    \label{tab:bench-grid-nd}
\end{table}

Next, we tested our implementation for \textsc{Network Diversion} on randomly generated Delaunay graphs.
They are planar by construction, and additionally, they provide diverse structures, capturing a wide variety of graph topologies. They are also relatively dense, which allows us to test the algorithm's performance on challenging instances.
The results are shown in \Cref{fig:nd-bench-delauay}.
We found that for a graph with~$n$ vertices, $\nicefrac{n \log n}{100}$ milliseconds was a very precise estimate for the median running time.

Finally, we compared our running time on randomly generated grid graphs to the running times reported by Cullenbine et al.~\cite{cullenbine2013theoreticalcomputational} on similarly generated grid graphs.
Our running times are reported in \Cref{tab:bench-grid-nd} and the results are compared to the results by Cullenbine et al.\ in~\Cref{fig:compare}.
The grids were generated by first creating an $N \times N$ grid, then uniformly at random select $s$, $t$, and $b$, and add random weights to all edges.  This matches the method of
Cullenbine et al.~\cite{cullenbine2013theoreticalcomputational}, with the exception that they always choose $s$ and $t$ to be on the outer face.
Running our algorithm with $s$ and $t$ chosen to be on the outer face performs at average slightly faster than $t_{50}$ reported in \Cref{tab:bench-grid-nd}.
On grids of size larger than $100 \times 100$, Cullenbine et al.\ reported timeouts and for grids of size at most $100 \times 100$, our algorithm is roughly 1000 times faster.  For reference, a line $c n \log n$---for some constant~$c$ that is chosen to fit our curve---is also shown.
We note that we only provide the running times as Cullenbine et al.\ reported, and that they use \textsc{ampl/cplex}~12.1 on a Linux machine with four 2.27--GHz processors, whereas ours is run on a  laptop on a single 2.80--GHz core.

\begin{figure}
    \centering
    \includegraphics[width=\linewidth]{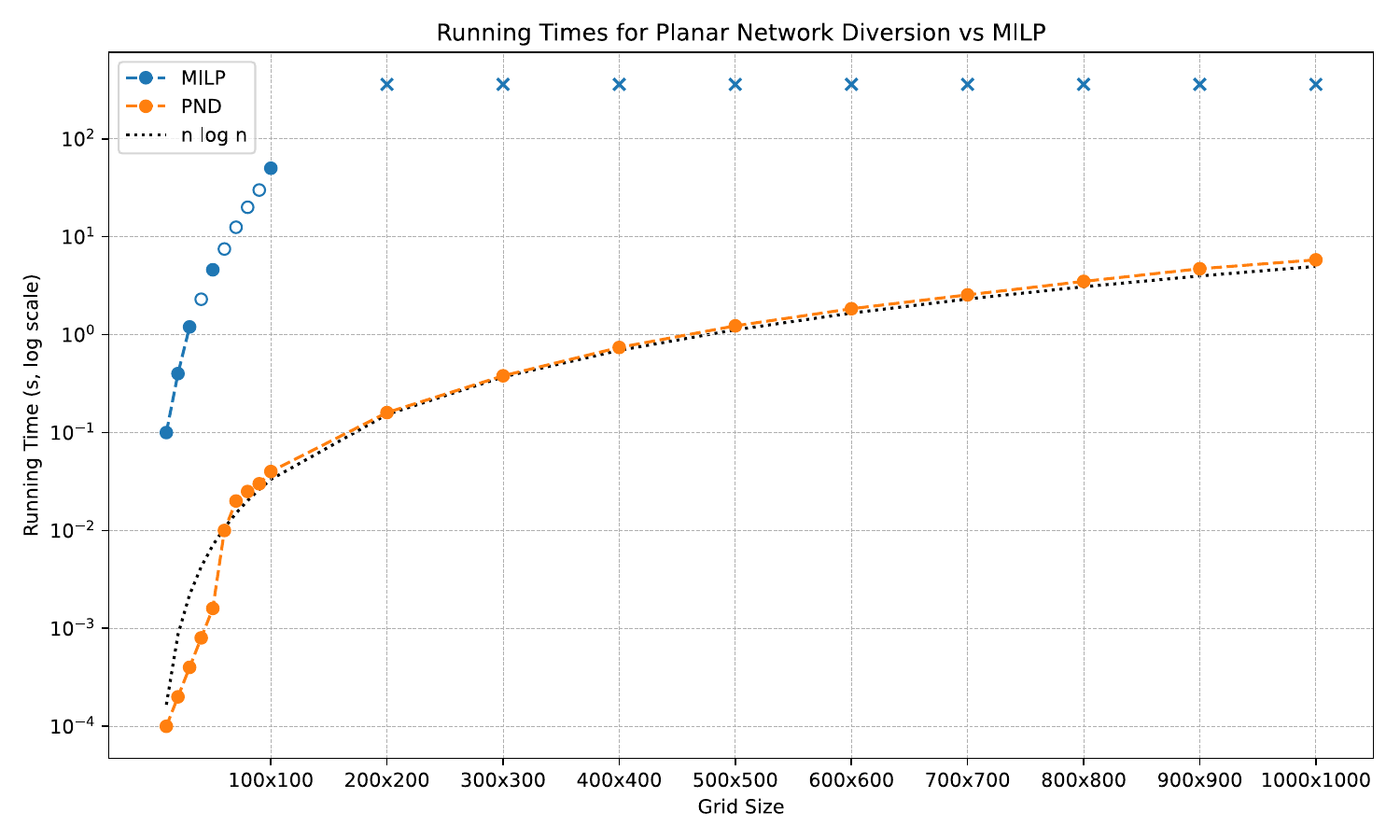}
    \caption{Comparison of the previous \textsc{milp} algorithm by Cullenbine et al.~\cite{cullenbine2013theoreticalcomputational} (blue, with filled circles for the values they provide and empty circles for interpolated values, and crosses for timeouts) and ours (orange) on grids of different sizes.
    }
    \label{fig:compare}
\end{figure}

In their experiments, the grids of size~$100 \times 100$ ran in roughly 50 seconds\footnote{They report $\mu=50.1~\text{s}$, $\sigma=39.6~\text{s}$ for 10 runs.  Our algorithm runs in time $\mu=0.0341~\text{s}$, $\sigma=0.0061~\text{s}$.} and larger grids timed out.
Our algorithm solved grids of size~$100 \times 100$ in less than 0.05 seconds which is roughly 100 times faster than the MILP.
Our algorithm was able to solve grids of size~$2000 \times 2000$ in less than 30 seconds.

\section{Conclusion}
\label{sec:conclusion}

We developed a simple yet efficient~$O(n \log n)$-time algorithm for \textsc{Network Diversion} for planar graphs.
We implemented the algorithm in Rust and it performs extremely well, being able to solve instances with millions of nodes in less than 30 seconds.

This is the first deterministic algorithm for \textsc{Planar Network Diversion} and the only algorithm that handles edge weights.  The algorithm is significantly simpler than the previous known algorithm by Bentert et al.~\cite{bentert2024twosets}.

The complexity of \textsc{Network Diversion} on general graphs is still open, and is a very interesting problem. Indeed, our technique does not even apply to graphs of genus one.
However, we conjecture that the problem remains polynomial-time solvable on these graphs.
Another interesting question is parameterization by \emph{crossing number} as many road and transportation networks are not planar due to the existence of bridges and underground tunnels.  However, it is often safe to assume that these graphs can be drawn in the plane with few crossing edges.
Is it possible to solve \textsc{Network Diversion} in time $f(c) \cdot n^{O(1)}$ for some function~$f$, where~$c$ is the crossing number of the given graph?

\bigskip
\noindent
\textbf{Acknowledgments.} We sincerely thank Petr Golovach and Tuukka Korhonen for their insightful discussions in the early stages of this work, which helped shape the development of the algorithm presented in this paper.

%%%%%%%%%%%%%%%%%%%%%%%%%%%%%%%%%%%%%%%
%%%%%%%%%%%%%%%%%%%%%%%%%%%%%%%%%%%%%%%

%%%%%%%%%%%%%%%%%%%%%%%%%%%%%%%%%%%%%%%%%
%%%%%%%%%%%%%%%%%%%%%%%%%%%%%%%%%%%%%%%%%

\clearpage

\appendix

\section{Illustration of diverse cuts}

Diverse minimal $s$-$t$-cuts computed on a Delaunay graph with 35 vertices and 91 edges with fixed~$s$ and~$t$.
The weights on the edges were chosen to be inverse proportional~$1/\ell$ to its length.
For each edge $e \in E(G)$, we run \textsc{Network Diversion} with input $(G, s, t, e)$.
Several instances yield the same solution, resulting in significantly fewer distinct solutions (34) than the 91 edges in the graph, which one might initially expect to correspond to the number of computed cuts.
Most cuts were found twice with one cut being found nine times.
Here, we have plotted all unique solutions, sorted by cost, with the cheapest being the true minimum $s$-$t$-cut.
Total computation time was around 0.5 seconds. (Total elapsed time was 469 ms, with each computation taking $5.162 \pm 2.79$~ms, incl.\ call via \texttt{subprocess.run} from Python.)

\begin{figure}[h!]
    \centering
    \includegraphics[width=0.82\textwidth]{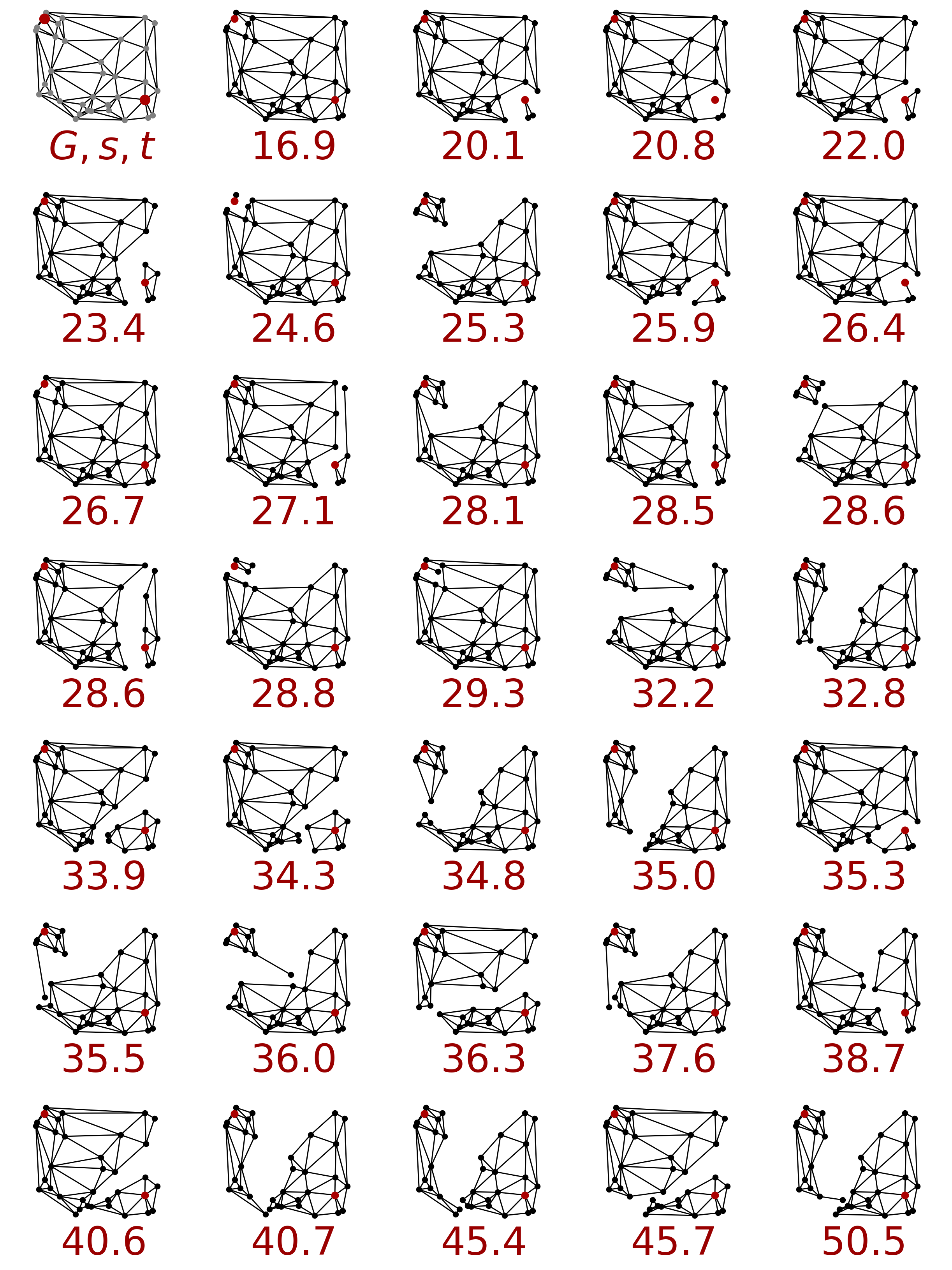}
    \caption{Diverse minimal $s$-$t$-cuts, with associated costs, computed on a Delaunay graph with 35 vertices and 91 edges.
    }
    \label{fig:diverse-cuts}
\end{figure}

\end{document}